\documentclass[11pt,a4paper]{article}
\usepackage[hyperref]{acl2019}
\usepackage{times}
\usepackage{latexsym}
\usepackage{url}
\usepackage{mathptmx}
\usepackage{graphicx}
\usepackage{mathtools,amssymb}
\usepackage{float}
\usepackage{amsthm}
\usepackage{breqn}
\usepackage{tcolorbox}
\usepackage{tikz}
\usetikzlibrary{arrows,automata,shapes,positioning,calc} 
\usepackage{tikz-qtree}
\usepackage{enumitem}
\usepackage{algorithm2e}
\usepackage{wrapfig}
\usepackage{graphicx,caption}
\usepackage{tipa}

\renewcommand{\succ}{\ensuremath\lhd}

\colorlet{good}{blue!75}
\colorlet{bad}{red!75}

\newcommand{\Pref}{\ensuremath{\mathtt{Pref}}}
\newcommand{\Suff}{\ensuremath{\mathtt{Suff}}}
\newcommand{\Substr}{\ensuremath{\mathtt{Substr}}}
\newcommand{\Subfact}{\ensuremath{\mathtt{Subfact}}}
\newcommand{\Supfact}{\ensuremath{\mathtt{Supfact}}}
\newcommand{\NSupFacs}{\ensuremath{\mathtt{NextSupFact}}}

\newcommand{\tuple}[1]{\ensuremath{\langle #1 \rangle}}

\newcommand{\defeq}{\overset{def}{=} }

\newtheorem{theorem}{Theorem}
\newtheorem{definition}{Definition}
\newtheorem{remark}{Remark}
\newtheorem{lemma}{Proposition}

\aclfinalcopy 

\title{Learning with Partially Ordered Representations}

\author{Jane Chandlee \\
	Tri-Co Department of Linguistics \\
	Haverford College\\
	\texttt{jchandlee@haverford.edu} \\\And
	R{\'e}mi Eyraud \\
	QARMA team, LIS \\
	Aix-Marseille University \\
	\texttt{remi.eyraud@lis-lab.fr} \\\AND
	Jeffrey Heinz \\
	Department of Linguistics \\
	Institute for Advanced Computational Science\\
	Stony Brook University\\
	\texttt{jeffrey.heinz@stonybrook.edu}\\\And
	Adam Jardine \\
	Department of Linguistics \\
	Rutgers University \\
	\texttt{adam.jardine@rutgers.edu} \\\AND
	Jonathan Rawski \\
	Department of Linguistics \\
	Institute for Advanced Computational Science\\
	Stony Brook University\\
	\texttt{jonathan.rawski@stonybrook.edu}}

\date{}

\begin{document}
	\maketitle
	\begin{abstract}
		This paper examines the characterization and learning of grammars defined with enriched representational models. 
		Model-theoretic approaches to formal language theory traditionally assume that each position in a string belongs to exactly one unary relation.
		We consider unconventional string models where positions can have multiple, shared properties, which are arguably useful in many applications.
		We show the structures given by these models are partially ordered, and present a learning algorithm that exploits this ordering relation to effectively prune the hypothesis space. 
		We prove this learning algorithm, which takes positive examples as input, finds the most general grammar which covers the data. 
	\end{abstract}

\section{Introduction} 

Foundational connections between formal languages, finite-state automata, and logic have been known for decades \citep{Buchi1960,Thomas1997}. Logical approaches are advantageous since they flexibly admit different representations. In many domains, such as biological sequencing or linguistics, shared properties of symbols in sequences provide information currently ignored by string-based inference algorithms, which largely focus on learning automata \citep{Higuera2010}. Here we explore the idea that domain-specific knowledge can be encoded representationally via model theory \citep{Libkin2004}, and shows how these representations can facilitate pattern learning. 

This paper synthesizes results in grammatical inference and model theory to present a novel algorithm which learns classes of formal languages using enriched representations of strings. In fact, our model-theoretic approach immediately generalizes these results to arbitrary data structures. 
Here we are concerned with the learning of those formal languages which can be defined via a set of structural constraints, such as the Strictly $k$-Local and Strictly $k$-Piecewise languages \citep{RogersPullum2011,Rogers-HeinzEtAl-2010-LPTSS}. Models of strings in the languages must not contain these forbidden structures \citep{Rogers-HeinzEtAl-2013-CSC}. 
Specifically, we define a learner whose hypothesis space is structured as a partial order by the relational signature of the particular model theory. 
We show how to traverse this space bottom-up from positive data to find a grammar which covers the data with the most general constraints.

\begin{figure*}[htbp!]
	\centering
	\begin{tikzpicture}[shorten >=1pt,->,thick, scale=0.8]
	\tikzstyle{vertex}=[circle,fill=black!25,minimum size=17pt,inner sep=0pt]
	\draw node[vertex] [label={[align=center]above:a}] (1) at (0,0) {1};
	\draw node[vertex][label={[align=center]above:b}] (2) at (2,0) {2};
	\draw node[vertex][label={[align=center]above:b}] (3) at (4,0) {3};
	\draw node[vertex][label={[align=center]above:a}] (4) at (6,0) {4};
	\draw (1) -- (2); \node at (1,0.25) {$\triangleleft$};
	\draw (2) -- (3); \node at (3,0.25) {$\triangleleft$};
	\draw (3) -- (4); \node at (5,0.25) {$\triangleleft$};
	
	\draw node[vertex][label={[align=center]above:a}] (s) at (8,0) {1};
	\draw node[vertex][label={[align=center]above:b}] (r) at (10,0) {2};
	\draw node[vertex][label={[align=center]above:b}] (i) at (12,0) {3};
	\draw node[vertex][label={[align=center]above:a}] (S) at (14,0) {4};
	\draw (s) -- (r); \node at (9,0.25) {$<$};
	\draw (r) -- (i); \node at (11,0.25) {$<$};
	\draw (i) -- (S); \node at (13,0.25) {$<$};
	\draw (s) .. controls +(+30:2cm) and +(+150:2cm) .. (i);
	\draw (s) .. controls +(+40:2.5cm) and +(+140:2.5cm) .. (S);
	\draw (r) .. controls +(-30:2cm) and +(-150:1cm) .. (S);
	\node at (10.5,.9)  {$<$};
	\node at (11,1.5)   {$<$};
	\node at (12,-0.75) {$<$};
	\end{tikzpicture}
	\caption{Visualizations of the successor (left) and precedence (right) models of $abba$.}
	\label{fig:cm}
\end{figure*}
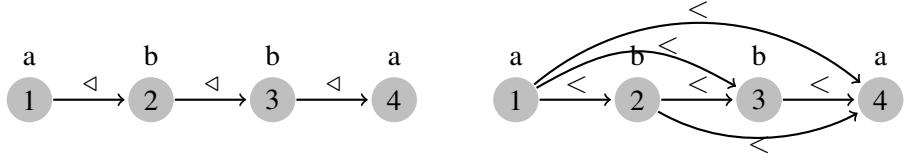

The paper is structured as follows:
Section 2 provides mathematical preliminaries in model theory. 
Section 3 characterizes ordering relations over these structures.
Section~4 generalizes the grammars employed in string extension and lattice-based learning \citep{Heinz-2010-SEL,Heinz-KasprzikEtAl-2012-LLHS} to show how these model theoretic structures can define classes of formal languages.
Section 5 discusses some entailments our learning algorithm takes advantage of. 
Section 6 defines a learning problem and criteria for selecting adequate solutions. 
Section 7 presents a general-to-specific, bottom-up algorithm which provably satisfies the learning criteria. 
Section 8 concludes the paper.  

\section{Preliminaries}
\subsection{Elements of Language Theory}
\label{subsec:elem}

The set of all possible finite strings of symbols from a finite
alphabet $\Sigma$ and the set of strings of length $\leq n$ are
$\Sigma^*$ and $\Sigma^{\leq n}$, respectively. The unique empty
string is represented with $\lambda$. The length of a string $w$ is
$|w|$, so $|\lambda| =$ 0. 
If $u$ and $v$ are two strings then we denote their concatenation
with $uv$. If $w$ is a string and $\sigma$ is the $i$th symbol in $w$ then $w_i=\sigma$, so $abcd_2=b$.


The set of prefixes of $w$, $\Pref(w)$, is
$\{p \in \Sigma^*\mid(\exists s \in \Sigma^*)[w = ps]\}$, the set
of suffixes of $w$, $\Suff(w)$, is $\{s \in \Sigma^*\mid(\exists p \in
\Sigma^*)[w = ps]\}$, the set of substrings, $\Substr(w)$, is $\{u\in \Sigma^*\mid(\exists l, r \in \Sigma^*)[w=lur] \}$, and the set of subsequences, $\mathtt{Subseq}(w) = 
{u_1u_2\cdots u_n|
	\exists v_0\cdot v_1\cdots v_n\in\Sigma^*[w=v_0u_1v_1\cdots u_nv_n]}
$

\subsection{Elements of Finite Model Theory}

Model theory, combined with logic, provides a powerful way to study and understand mathematical objects with structures~\citep{Enderton2001}.  
In this paper we only consider finite relational models~\citep{Libkin2004} of strings in $\Sigma^*$.
\begin{definition}[Models]
	A \emph{model signature} is a tuple $S=\tuple{D;R_1,R_2,\ldots,R_m}$ where the domain $D$ is a finite set, and each $R_i$ is a $n_i$-ary relation over the domain. 
	A \emph{model for a set of objects} $\Omega$ is a total, one-to-one function from $\Omega$ to structures whose type is given by a model signature.
\end{definition}
For example, a conventional model for strings in $\Sigma^*$ is given by the signature $\Gamma^\succ\defeq\tuple{D;\succ, [R_\sigma]_{\sigma\in\Sigma} }$ and the function $M^\succ: \Sigma^*\to \Gamma^\succ$ such that $M^\succ(w) \defeq \tuple{D^w;\succ, [R^w_\sigma]_{\sigma\in\Sigma} }$
where $D^w\defeq \{1,\ldots, |w|\}$ is the domain, $\succ\defeq\{(i,i+1)\in D\times D \}$ is the successor relation which orders the elements of the domain, and $[R^w_\sigma]_{\sigma\in\Sigma}$ is a set of $|\Sigma|$ unary relations such that for each $\sigma\in\Sigma$, $R^w_\sigma\defeq\{ i\in D^w\mid w_i=\sigma\}$.
We will usually omit the superscript $w$ since it will be clear from the context.

For example, with $\Sigma=\{a,b,c\}$ and the model above for strings, we have 
$M^\succ(abba)= \big\langle D=\{1, 2, 3, 4\}; \succ=\{(1, 2), (2,3), (3,4)\}, R_a=\{1,4\}, R_b=\{2,3\}, R_c=\emptyset\big\rangle~.$

Figure~\ref{fig:cm} illustrates $M^\succ(abba)$ on the left.




Another conventional model is the precedence model, with the signature $\Gamma^< \defeq\tuple{D;<, [R_\sigma]_{\sigma\in\Sigma} }$. 
It differs from the successor model only in that the order relation is defined with general precedence $< \defeq\{(i,j)\in D\times D\mid i < j\}$  \citep{Buchi1960,McNaughtonPapert1971,Rogers-HeinzEtAl-2013-CSC}. 
Under this signature, the string $abba$ has the following model.

$M^<(abba) =  \big\langle D=\{1,2,3,4\}; <=\{(1,2),(1,3),(1,4), (2,3), (2,4), (3,4)\}, R_a=\{1,4\}, R_b=\{2,3\}, R_c=\emptyset\big\rangle$.

Figure~\ref{fig:cm} illustrates $M^{<}(abba)$ on the right.


The model-theoretic framework is not unique to strings. It extends to arbitrary data structures by expanding parts of the model signature. For example, \citet{rogers2003syntactic} describes a model-theoretic characterization of trees of arbitrary dimensionality where the domain $D$ is a Gorn tree domain \cite{gorn1967explicit}.  This is a hereditarily prefix closed set \textit{D} of node addresses, that is to say, for every $d \in D$ with $d = \alpha i$, where $i \in \mathbb{N}$, $\alpha \in \mathbb{N}^*$  it holds that $\alpha \in D$, and for every $d \in D$ with $d = \alpha i \neq \alpha0,$ then $\alpha(i-1) \in D$. 

In this view, a string may be called a one-dimensional or unary-branching tree, since it has one axis along which its nodes are ordered. In a standard tree mdoel signature, the set of nodes is ordered by two binary relations, ``dominance" and ``immediate left-of". Suppose $s$ is the mother of two nodes $t$ and $u$ in some standard tree, and also assume that $t$ precedes $u$. Then we might say that $s$ dominates the string $tu$. Standard or two-dimensional trees, then, relate nodes to one-dimensional trees (strings) by immediate dominance. A three-dimensional tree relates nodes to two-dimensional, i.e. standard trees, corresponding to Tree-Adjoining Grammar derivations. In general, a $d$-dimensional tree is a set of nodes ordered by $d$ dominance relations such that the $n$-th dominance relation relates nodes to $(n-1)$-dimensional trees (for $d = 1$, single nodes are zero-dimensional trees).


While a Gorn tree domain as written encodes these dominance and precedence relations implicitly, we may explicitly write them out model-theoretically so that a signature for a $\Sigma$-labeled 2-$d$ tree is $\Gamma^{\succ\prec}=\tuple{D;\succ,\prec, [R_\sigma]_{\sigma\in\Sigma} }$ where $\succ$ is the ``immediate dominance" relation and $\prec$ is the ``immediate left-of" relation. Model signatures that include transitive closure relations of each of these have also been studied.

\vspace{-.5cm}
\begin{figure}[hbp!]
	\centering
	\begin{tikzpicture}
	[level 1/.style={sibling distance=24mm,level distance=12mm},
	level 2/.style={sibling distance=12mm}]	
	\node {$\epsilon$}
	child {node  (D){0} 
		child {node (E){00}}
		child {node (F){01}
			child {node (G){010}}
			child {node (H){011}}
		}
	}
	child {node (I){1}
		child {node (J){10}}
		child {node (M){11}
			child {node (N){110}}
			child {node (O){111}
				child {node (L){1110}}
			}
			child {node (P){112}}
		}
	};
	
	\draw[dotted] (D) -- (I);
	\draw[dotted] (E) -- (F);
	\draw[dotted] (J) -- (M);
	\draw[dotted] (G) -- (H);
	\draw[dotted] (N) -- (O);
	\draw[dotted] (O) -- (P);
	
	\end{tikzpicture}

	\caption{2-dimensional tree model. Dominance and precedence relations shown with solid/dashed and dotted lines, respectively}
	\label{fig:3dtree}
\end{figure}
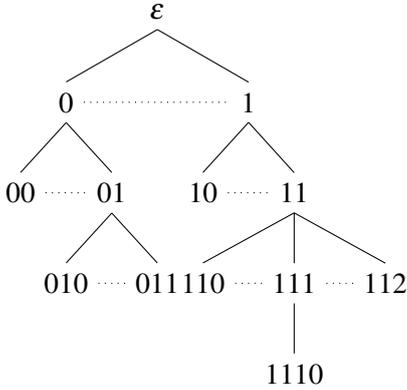


\subsection{Unconventional Word Models}\label{ss:uwm}

Whereas \citet{rogers2003syntactic} generalized conventional word models to trees, here we generalize word models in a different way. Conventional string models are the successor and precedence models introduced previously. 
What makes these models conventional is the unary relations which essentially label each domain element with a single, mutually exclusive, property: the property of being some $\sigma\in\Sigma$.

In contrast, unconventional models for strings recognize that distinct alphabetic symbols may share properties, and expands the model signature by including these properties as unary relations \citep{Strother-Garcia-HeinzEtAl-2016-UMTGICSP,VuEtAl2018}. 
For example, a conventional model of $\Sigma=\{\mathtt{a,\ldots,z,A,\ldots,Z}\}$ would include 52 unary relations, one for each lowercase and capital letter. 
On the other hand, an unconventional model might only include 27: 26 for the letters, and one unary relation  \texttt{Capital}. 
Then, letters \texttt{A} and \texttt{a} share the `A' property and \texttt{A} additionally has the property of being \texttt{Capital}. 

In linguistics, speech sounds are commonly decomposed into binary features based on their phonetic properties. 
So the set of segments $\{$\textipa{z,Z,d,b,g,\ldots}$\}$ all share the property \texttt{+Voice}, meaning the vocal cords are activated, 
while the segments $\{$\textipa{s,S,t,p,k,\ldots}$\}$ share the 
property \texttt{-Voice}, meaning the vocal cords are not activated. 
Thus unconventional models may refer to individual features in defining grammatical constraints, rather than each individual segment.
  
Different representations of strings and trees provide a unified perspective on well-known subclasses of the regular languages from a model-theoretic and logical perspective \citep{Thomas1997,Rogers-HeinzEtAl-2013-CSC}. 
However, they also open up new doors for grammatical inference by allowing one to consider other models for strings \citep{Strother-Garcia-HeinzEtAl-2016-UMTGICSP,VuEtAl2018}.

\section{Subfactors, Superfactors, Ideals and Filters}

We sometimes refer to the model of a string $w$ as a \emph{structure}. 
However, structures are more general in that they correspond to any mathematical structure conforming to the model signature. 
As such, while a model of a string $w$ will always be a structure, a structure will not always be a model of a string $w$. 
The \emph{size} of a structure $S$, denoted $|S|$, coincides with the cardinality of its domain. 

We next wish to introduce a partial ordering over structures. To do so, we must define the terms \emph{connected, restriction}, and \emph{factor}. 
For each structure
$S=\tuple{D;\succ, R_1, \ldots R_n}$ let the binary ``connectedness'' relation $C$ be
defined as follows.

$  C  \defeq  \big\{(x,y)\in D\times D \mid
\exists i\in\{1\ldots n\}, 
\exists(x_1\ldots x_m)\in R_i, 
\exists s,t\in\{1\ldots m\}, 
x=x_s, y=x_t\big\}$

Informally, domain elements $x$ and $y$ belong to $C$ provided they belong to some non-unary relation. Let $C^*$ denote the symmetric transitive closure of $C$.

\begin{definition}[Connected structure]
	A structure $S=\tuple{D;\succ, R_1,R_2,\ldots,R_n}$ is \emph{connected} iff for all $x,y\in D$, $(x,y)\in C^*$.
\end{definition}
For example, $M^\succ(abba)$ above is a connected structure. However, the structure $S_{ab,~ba}$ shown below which is identical to $M^\succ(abba)$  except it omits the pair (2,3) from the order relation is not connected since none of (1,3),(1,4), (2,3) nor (2,4) belong to $C^*$.  
$S_{ab,~ba}= \big\langle D=\{1,2,3,4\}; \succ=\{(1,2),(3,4)\}, R_a=\{1,4\}, R_b=\{2,3\}, R_c=\emptyset\big\rangle$

\begin{center}
	\begin{tikzpicture}[shorten >=1pt,->,thick]
	\tikzstyle{vertex}=[circle,fill=black!25,minimum size=17pt,inner sep=0pt]
	\draw node[vertex] [label={[align=center]above:a}] (1) at (0,0) {1};
	\draw node[vertex][label={[align=center]above:b}] (2) at (2,0) {2};
	\draw node[vertex][label={[align=center]above:b}] (3) at (4,0) {3};
	\draw node[vertex][label={[align=center]above:a}] (4) at (6,0) {4};
	\draw (1) -- (2); \node at (1,0.25) {$\triangleleft$};
	\draw (3) -- (4); \node at (5,0.25) {$\triangleleft$};
	\end{tikzpicture}
\end{center}

Note that no string in $\Sigma^*$ has structure $S_{ab,~ba}$ as its model.

\begin{definition}
	\label{def:restriction}
	$A=\tuple{D^A; \succ, R_1^A, \ldots, R_n^A}$ \emph{is a restriction of} $B=\tuple{D^B;\succ, R_1^B, \ldots, R_n^B}$ iff $D^A\subseteq D^B$ and for each $m$-ary relation $R_i$, we have
	$R_i^A=\{(x_1\ldots x_m)\in R_i^B\mid x_1, \ldots, x_m\in D^A\}$.
\end{definition}   
Informally, one identifies a subset $A$ of the domain of $B$ and strips $B$ of all elements and relations which are not wholly within $A$. 
What is left is a restriction of $B$ to $A$. 
\begin{definition}
	Structure $A$ is a \emph{subfactor} of structure $B$ ($A\sqsubseteq B$) if $A$ is connected,  there exists a restriction of $B$ denoted $B'$, and there exists $h:A\rightarrow B'$ such that for all $a_1,\ldots a_m\in A$ and for all $R_i$ in the model signature: if $h(a_1),\ldots h(a_m)\in B'$ and $R_i(a_1,\ldots a_m)$ holds in $A$ then $R_i(h(a_1),\ldots h(a_m))$ holds in $B'$. If $A\sqsubseteq B$ we also say that $B$ is a \emph{superfactor} of $A$.
\end{definition}
In other words, properties that hold of the connected structure $A$ also hold in a related way within $B$. 

If $A\sqsubseteq B$ and $|A|=k$ then we say $A$ is a $k$-subfactor of $B$. 
For all $w\in\Sigma^*$, and for any model $M$ of $\Sigma^*$, let the subfactors of $w$ be $\Subfact(M,w) = \{A\mid A\sqsubseteq M(w)\}$ and the $k$-subfactors of $w$ be $\Subfact_k(M,w) = \{A\mid A\sqsubseteq M(w),~ |A|\leq k\}$. 
We also define $\Subfact(M,\Sigma^*)$ to be $\bigcup_{w\in\Sigma^*} \Subfact(M,w)$ and $\Subfact_k(M,\Sigma^*)$ to be $\bigcup_{w\in\Sigma^*} \Subfact_k(M,w)$. 
When $M$ is understood from context, we write $\Subfact(w)$ instead of $\Subfact(M,w)$. 
We define the sets of superfactors   $\Supfact(M,w)$ and $\Supfact(M,\Sigma^*)$ similarly. 

Observe that $(\Subfact(M,w), \sqsubseteq)$ is a partially ordered set (poset). The next definition and lemma establishes that models of strings are principal elements of ideals and filters.

\begin{definition}[Ideals]
	A subset $I$ of a poset is an Ideal if 
	\begin{itemize}
		\item $I$ is non-empty
		\item for every $x$ in $I$, $y \leq x$ implies that $y$ is in $I$
		\item for every $x,y$ in $I$, there exists some element $z$ in $I$, such that $x \leq z$ and $y \leq z$.
	\end{itemize}
\end{definition}
The dual of an ideal is a filter.
\begin{definition}[Filters]
	A subset $F$ of a poset is a \emph{filter} iff 
	\begin{itemize}
		\item $F$ is non-empty
		\item for every $x$ in $F$, $x \leq y$ implies that $y$ is in $F$
		\item for every $x,y$ in $F$, there exist some element $z$ in $F$, such that $z \leq x$ and $z \leq y$.
	\end{itemize}
\end{definition}

\begin{definition}[Principal Ideals, Filters and Elements]
	For any poset  $\langle X, \leq \rangle$, the smallest filter containing $x\in X$ is a \emph{principal filter} and $x$ is the \emph{principal element} of this filter. 
	Similarly, the smallest ideal containing $x\in X$ is a \emph{principal ideal} and $x$ is the \emph{principal element} of this ideal.
\end{definition}

\begin{remark}
	Given a model $M$ of $\Sigma^*$ and $k>0$, $\Subfact_k(M,w)$ is a principal ideal in $\Subfact(M,\Sigma^*)$ whose principal element is $M(w)$.
	$\Supfact_k(M,w)$ is a principal filter in $\Supfact(M,\Sigma^*)$ whose principal element is $M(w)$.
	The empty structure $\tuple{\emptyset; \emptyset, \ldots \emptyset}$ is a subfactor of every structure in $\Subfact(M,\Sigma^*)$. 
\end{remark}


The next two propositions show how this representational perspective unifies the treatment of substrings and subsequences. They are subfactors under the successor and precedence models, respectively. 
A string $x=x_1\cdots x_n$ is a substring of $y$ iff there exists $l, r$ such that $y=lxr$. String $x$ is a subsequence of $y$ iff there exists $v_0, v_1, \ldots v_n$ such that $w=v_0x_1v_1\cdots x_nv_n$.
\begin{lemma}[Substrings are subfactors under $M^\succ$]
	\label{lemma:sub-fact}
	For all strings $x,y\in\Sigma^*$, $x$ is a substring of $y$ iff $M^\succ(x)\sqsubseteq M^\succ(y)$.
\end{lemma}
\begin{proof}
	Note that the result trivially holds for $x=\lambda$: we restrict ourselves to the case $x\neq \lambda$.
	Let $M^\succ(x)=\tuple{D^x; \succ, [R^x_\sigma]}$ and $M^\succ(y)=\tuple{D^y; \succ, [R^y_\sigma]}$
	
	($\Rightarrow$). Suppose $x$ is a substring of $y$: it exists $l, r$ such that $y = lxr = \sigma_1 \ldots \sigma_{|l|}\sigma_{|l|+1} \ldots \sigma_{|l|+|x|}\sigma_{|l|+|x|+1} \ldots \sigma_{|l|+|x|+|r|}$.
	This implies that, for all $i$, $1\leq i\leq |x|$, $d \in R^y_{\sigma_{|l|+i}}$ iff $d \in R^x_{\sigma_{i}}$. Thus, if we set the isomorphism $\phi$ to be such that $\phi(i)=|l|+i$ for $1\leq i \leq |x|$, we have $\phi(M^\succ(x))$ that is a restriction of $M^\succ(y)$, and therefore $M^\succ(x)\sqsubseteq M^\succ(y)$ by definition.
	
	($\Leftarrow$). Let $y$ be the sequence of letters $\sigma_1\ldots \sigma_{|y|} $ and suppose $M^\succ(x)\sqsubseteq M^\succ(y)$: there exists a isomorphism $\phi: \{1, \ldots, |x|\} \to \{1, \ldots, |y| \}$ such that $\phi(M^\succ(x))$ is a restriction of $M^\succ(y)$. This means that $\phi(D^x)\subseteq D^y$ and for all $\sigma$: $\phi(R^x_\sigma)=\{\phi(i) \in R^y_\sigma\mid \phi(i) \in \phi(D^x)\}$ (Definition~\ref{def:restriction}). 
	This implies that $x=\sigma_{\phi(1)}\ldots\sigma_{\phi(|x|)}$.
	Given that $\succ = \{(i, i+1)\in D\times D\}$,
	we have $\phi(i+1)=\phi(i)+1$ and thus there exist $l$ and $r$ in $\Sigma^*$ such that $y = l\sigma_{\phi(1)}\ldots\sigma_{\phi(|x|)} r = lxr$.
\end{proof}

\begin{lemma}[Subsequences are subfactors under $M^<$]
	For all strings $x,y\in\Sigma^*$, $x$ is a subsequence of $y$ iff $M^<(x)\sqsubseteq M^<(y)$.
\end{lemma}
\begin{proof}
	We leave this proof to the Reader since it is of similar nature to the previous one.
\end{proof}

\section{Grammars, Languages, and Language Classes}
\label{sec:grammars}
Factors can define grammars, formal languages, and classes of formal languages. 
Usually a model signature provides the vocabulary for some logical language. 
Sentences in this logical language define sets of strings as follows.
The language of a sentence $\phi$ is all and only those strings whose models satisfy $\phi$. 
Within the regular languages, many well-known subregular classes can be characterized logically in this way \citep{McNaughtonPapert1971,RogersPullum2011,Rogers-HeinzEtAl-2013-CSC,Thomas1997}. 

Intuitively, the grammars we are interested in consist of a finite list of \emph{forbidden} subfactors, whose largest size is bounded by $k$. 
Strings in the language of this grammar are those which do not contain any forbidden subfactors. 
In this way these grammars are like logical expressions which are "conjunctions of negative literals" \citep{Rogers-HeinzEtAl-2013-CSC} where the negative literals are played by the the forbidden factors.

Each forbidden subfactor is a principal element of a filter and the language is all strings whose models are not in any of these filters.
For each $k$, there is a class of languages including all and only those languages that can be defined in this way. 
For example, the Strictly $k$-Local (SL$_k$) and Strictly $k$-Piecewise languages can be defined in this way; they are languages which forbid finitely many substrings or subsequences, respectively \citep{GarciaEtAl1990,Rogers-HeinzEtAl-2010-LPTSS}. Formally:

\begin{definition}
	\label{def:language}
	Let $k$ be some positive integer, and $M$ a model of $\Sigma^*$ with signature $\Gamma$. A \emph{grammar} $G$ is a subset of $\Subfact_k(M,\Sigma^*)$. 
	The \emph{language} of $G$ is $L(G) = \{w\in\Sigma^*\mid \Subfact_k(M,w)\cap G =\emptyset\}$.
	The \emph{class of languages} $\mathcal{L}(M,k) =\{ L\mid \exists G \subseteq\Subfact_k(M,\Sigma^*), L(G) = L \}$.
\end{definition}
The elements of $G$ are principal elements of filters, and are called forbidden subfactors.

As an example, let $\Sigma=\{a,b,c\}$ and consider $G=\{M^\succ(aa),M^\succ(bb),M^\succ(c)\}$. 
$L(G)$ includes the strings $(ab)^+$ and  $(ba)^+$ and no other strings, because the substrings $aa$, $bb$, and $c$ are all forbidden. 
This language belongs to $\mathcal{L}(M^\succ,2)$.

\begin{lemma}\label{zeroInt}
	For each $w\in L(G)$ and each $g\in G$,                                                                                                                        $\Subfact(M,w)$ has a zero intersection with $\Supfact(g)$.
\end{lemma}
\begin{proof}
	Suppose there exists $A\in \Subfact_k(\Sigma^*)$ such that $A\sqsubseteq M(w)$ and $g\sqsubseteq A$. This implies that $g\sqsubseteq M(w)$ and thus that $\Subfact_k(M,w)\cap G\neq \emptyset$ which contradicts Definition~\ref{def:language}.
\end{proof}
In other words, the principal ideal of $M(w)$ is disjoint from the principal filters of the elements of $G$.

\section{Grammatical Entailments}

\begin{figure*}[htbp!]
	\centering
	
	\begin{tikzpicture}[scale=0.9]
	
	\node (s) at (0,0) {[]}; 
	\node (sc) at (-7em,4em) {[\texttt{capital}]};
	\node (sa) at (0,4em){[\texttt{a}]};
	\node (sb) at (6em,4em) {[\texttt{b}]};
	\node (sz) at (9em,4em) {\ldots};
	\node (sca) at (-4em,8em) {[\texttt{capital, a}]};
	\node (sbd) at (4em,8em) {[\texttt{a}][]};
	\node (scad) at (-10em,12em) {[\texttt{capital, a}][]};
	\node (sbad) at (1em,12em) {[\texttt{a}][\texttt{capital, a}]};
	\node (sdad) at (9em,12em) {[\texttt{a}][\texttt{a}]};
	\node (scada) at (-10em,16em) {\ldots};
	\node (sbadb) at (1em,16em) {\ldots};
	\node (sdadc) at (9em,16em) {\ldots};
	
	\foreach \source/\target in {s/sa,s/sb,s/sc,s/sz,sc/sca,sa/sca,sa/sbd,sca/scad,sca/sbad,sbd/scad,sbd/sbad,sbd/sdad,scad/scada,sbad/sbadb,sdad/sdadc}
	\draw (\source) to (\target);
	
	\node at (sa.north west) {\color{good}\checkmark};
	\node at (sc.north west) {\color{good}\checkmark};
	\node at (sca.north west) {\color{good}\checkmark};
	\node (check1) at (sa.north west) {\color{good}\checkmark};
	\node (check2) at (s.north west) {\color{good}\checkmark};
	\node (bad2) at (sbd.north west) {\color{bad}X};
	\node (bad2) at (sbad.north west) {\color{bad}X};
	\node (bad2) at (scad.north west) {\color{bad}X};
	\node (bad2) at (sdad.north west) {\color{bad}X};
	\node (bad2) at (scada.north west) {\color{bad}X};
	\node (bad2) at (sbadb.north west) {\color{bad}X};
	\node (bad2) at (sdadc.north west) {\color{bad}X};
	
	\draw [good,->,thick] ($(sca.south)+(-7em,.5em)$) -- +(0em,-8em);
	\draw [bad,->,thick] ($(sbd.north)+(7em,-.5em)$) -- +(0em,9em);
	
	\end{tikzpicture}
	\caption{The Structure ideals(blue) and filters(red) for a capitalized letter model.}
	\label{fig:poset}
\end{figure*}
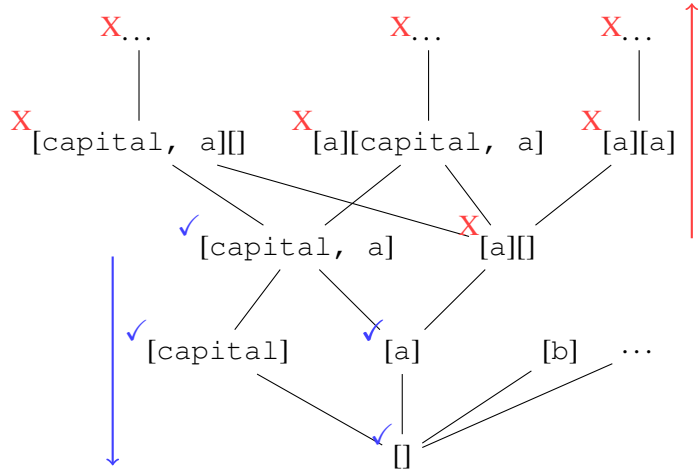

Given a grammar $G$, we call a subfactor $s$ in $\Subfact(\Sigma^*)$ \emph{ungrammatical} if it belongs to a principal filter of any element of $G$. Subfactors that are not ungrammatical are called \emph{grammatical}.
Lemma 14 ensures that grammaticality is downward entailing, in the sense that if a model of the word $M(w)$ is not contained in the principal filters of the elements of the grammar, then neither are the subfactors of $M(w)$. But it also ensures that ungrammaticality is upward entailing: if a model of the word $M(w)$ belongs to the principal filters of the elements of the grammar, then all of the superfactors of $M(w)$ in that filter are likewise contained.

In this way, the ideals and filters within a a particular model noted above give rise to these entailment properties of grammaticality with respect to the hypothesis space. If the learner constructs filters, then the grammar $G$ will allow structures such that language membership is downward entailing with respect to the grammar $G$, and language non-membership is upward entailing with respect to the grammar $G$.

\subsection{Example: Text Capitalization}

As an example, consider capitalized letters as discussed above. In an unconventional word model, each capital letter at some position $x$ is represented as satisfying one of the relations $R \in \{\mathtt{a}(x),\mathtt{b}(x),\ldots,\mathtt{z}(x)\}$ as well as the unary relation $\mathtt{capital}(x)$. 
Thus the relation $\mathtt{a}(x)$ is true of both lowercase $\mathtt{a}$ and uppercase $\mathtt{A}$, but $\mathtt{a}(x) \land \mathtt{capital}(x)$ is only true of uppercase $\mathtt{A}$. 
Note also that in this model no position $x$ of a structure can satisfy both predicates $\mathtt{a}(x)$ and $\mathtt{b}(x)$. We return to this point in \S\ref{sec:botuplearner}.


Figure 2 showcases the relationship among these structures under a model $M$. 
The structure for $\mathtt{A}$, $[\mathtt{capital},\mathtt{a}]$, contains as subfactors $[\mathtt{capital}]$, $[\mathtt{a}]$, [], and the empty structure (not shown). 
The empty structure is a subfactor of [], and [] in turn is a subfactor of $[\mathtt{capital}]$ and $[\texttt{a}]$. 
The subfactor $[\mathtt{a}]$ contains the subfactor [], the domain element with no relations, but has superfactors [\texttt{capital,a}], which has one domain element and two relations, and [a][], which has two domain elements, and the first satisfying the property \texttt{a}. 
Subfactors and superfactors are listed above and below each other, respectively, with lines between them. 
Members of one ideal are noted with a blue checkmark, and members of a filter are noted by a red asterisk.    

Applying this to the example in Figure~\ref{fig:poset}, if the structure $\mathtt{[capital,a]}$ is grammatical, then all of its subfactors, such as [\texttt{capital}] and [\texttt{a}], and [] are grammatical.
Since those are grammatical, each of their subfactors is also grammatical, which in this case is just [$\emptyset$], shown in blue in Figure~\ref{fig:poset}. 
Conversely, if the structure [\texttt{a}][] is known to be ungrammatical, then any structure which has it as a subfactor is also ungrammatical (in this example, [\texttt{capital,a}][], shown in Red in Figure~\ref{fig:poset}. To see the importance, consider a string with only lowercase letters. In a connected model, the grammar would ban 26 forbidden factors (A,B,C,...), but the ``capital" model bans just one, [\texttt{capital}].

\subsection{Example: Long Distance Linguistic Dependencies}

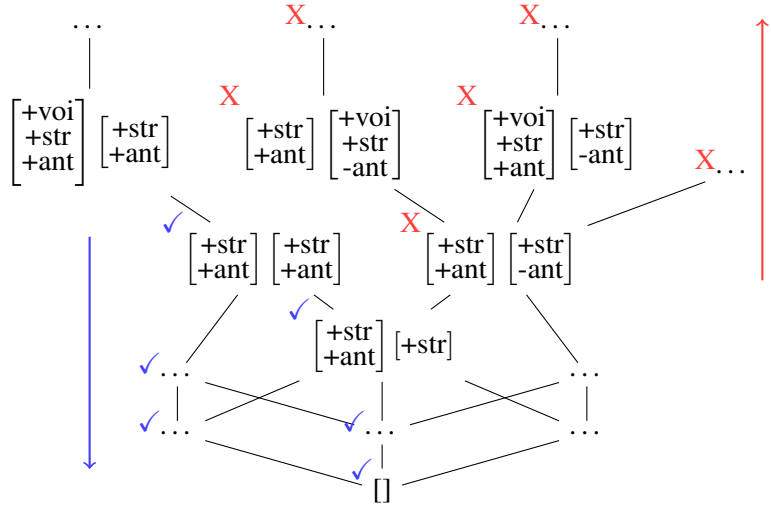
\begin{figure*}[htbp!]
	\centering
	\begin{tikzpicture}
	
	\node (0) at (0,0) {[]}; 
	\node (1) at (-7em,2em) {\ldots};
	\node (2) at (0em,2em) {\ldots};
	\node (3) at (7em,2em) {\ldots};
	
	\node (4) at (-7em,4em) {\ldots};
	\node (5) at (0,5em){$\left[\shortstack{+str\\+ant\\[-.7em]}\right]\left[\shortstack{+str}\right]$};
	\node (6) at (7em,4em) {\ldots};
	
	\node (7) at (-4em,8em) {$\left[\shortstack{+str\\+ant\\[-.7em]}\right]\left[\shortstack{+str\\+ant\\[-.7em]}\right]$};
	
	\node (8) at (4em,8em) {$\left[\shortstack{+str\\+ant\\[-.7em]}\right]\left[\shortstack{+str\\-ant\\[-.7em]}\right]$};

	\node (9) at (-10em,12em) {$\left[\shortstack{+voi\\+str\\+ant\\[-1em]}\right]\left[\shortstack{+str\\+ant\\[-.7em]}\right]$};
	
	\node (10) at (-2em,12em) {$\left[\shortstack{+str\\+ant\\[-1.5ex]}\right]\left[\shortstack{+voi\\+str\\-ant\\[-2.5ex]}\right]$};
	
	\node (11) at (6em,12em) {$\left[\shortstack{+voi\\+str\\+ant\\[-2.5ex]}\right]\left[\shortstack{+str\\-ant\\[-1.5ex]}\right]$};
	
	\node (12) at (-10em,16em) {\ldots};
	\node (13) at (-2em,16em) {\ldots};
	\node (14) at (6em,16em) {\ldots};
	\node (15) at (12em,11em) {\ldots};
	
	\foreach \source/\target in {1/5,1/4,2/5,2/4,2/6,3/5,3/6,6/8,4/7,5/7,5/8,7/9,8/10,8/11,10/13,11/14,8/15,9/12,0/1,0/2,0/3}
	\draw (\source) to (\target);
	
	\node at (1.north west) {\color{good}\checkmark};
	\node at (2.north west) {\color{good}\checkmark};
	\node at (5.north west) {\color{good}\checkmark};
	\node at (4.north west) {\color{good}\checkmark};
	\node at (7.north west) {\color{good}\checkmark};
	\node (check2) at (0.north west) {\color{good}\checkmark};
	\node (bad2) at (8.north west) {\color{bad}X};
	\node (bad2) at (10.north west) {\color{bad}X};
	\node (bad2) at (11.north west) {\color{bad}X};
	\node (bad2) at (13.north west) {\color{bad}X};
	\node (bad2) at (14.north west) {\color{bad}X};
	\node (bad2) at (15.north west) {\color{bad}X};
	
	\draw [good,->,thick] ($(7.south)+(-6em,2em)$) -- +(0em,-8em);
	\draw [bad,->,thick] ($(8.south)+(9em,.5em)$) -- +(0em,9em);
	
	\end{tikzpicture}
	\caption{Structure ideals(blue) and filters(red) for a phonological precedence model.}
	\label{fig:poset2}
\end{figure*}

As another example, sequences of speech sounds as mentioned earlier may be decomposed into binary features based on their phonetic properties like anteriority (\texttt{$\pm$ant} --- whether it occurs in the anterior of the vocal tract), stridency (\texttt{$\pm$str} --- whether it produces a high-intensity fricative noise), or voicing (\texttt{$\pm$voi} --- whether it activates the vocal chords), among others \citep{hayes2009}. Each sound at some position $x$ is represented as satisfying relations $R \in \{\mathtt{\pm voi}(x),\mathtt{\pm str}(x),\ldots,\mathtt{\pm ant}(x)\}$. 
Thus the relation $\mathtt{+str}(x)$ is true of both the sound \texttt{s} as in the first sound of ``sue" and \texttt{\textipa{S}}, as in ``shoe", but $\mathtt{+str}(x) \land \mathtt{-ant}(x)$ is only true of $\texttt{\textipa{S}}$. 

Note also that in this model no position $x$ of a structure can satisfy both predicates $\mathtt{+str}(x)$ and $\mathtt{-str}(x)$. We return to this point in \S \ref{sec:botuplearner} below. 
We again use square brackets to delimit the domain elements and write the unary features within them, so a model representation like \texttt{$\left[\shortstack{+str\\+ant\\[-.7em]}\right]\left[\shortstack{+str\\-ant\\[-.7em]}\right]$} has the following visual representation:

\begin{center}
	\begin{tikzpicture}[shorten >=1pt,->,thick]
	\tikzstyle{vertex}=[circle,fill=black!25,minimum size=17pt,inner sep=0pt]
	\draw node[vertex][label={[align=center]above:$\shortstack{+str\\+ant}$}] (s) at (0,0) {};
	\draw node[vertex][label={[align=center]above:$\shortstack{+str\\-ant}$}] (r) at (2,0) {};
	\draw (s) -- (r); \node at (1,0.25) {$<$};

	\end{tikzpicture}
\end{center}

To ease the exposition, we will use square brackets to delimit the domain elements and write the unary relations within them instead of specifying the model in mathematical detail. In an unconventional subsequence word model, then, one possible structure of the subsequence \texttt{s\ldots\textipa{S}} is written \texttt{$\left[\shortstack{+str\\+ant\\[-.7em]}\right]\left[\shortstack{+str\\-ant\\[-.7em]}\right]$}.        


In many languages, the presence of certain segments is dependent on the presence of another segment. In Samala, subsequences like \texttt{s\ldots s} are allowed but \texttt{s\ldots\textipa{S}} are not, so words like \texttt{ha\textbf{s}xintilawa\textbf{s}} are allowed but words like \texttt{ha\textbf{s}xintilawa\textbf{\textipa{S}}} are not \citep{Hansson2010}. In an unconventional model, banning structures of the form \texttt{[+str][+str]} is insufficient, since all these segments share that stridency property, while a structure like \texttt{$\left[\shortstack{+str\\+ant\\[-.7em]}\right]\left[\shortstack{+str\\-ant\\[-.7em]}\right]$} will distinguish them, since they disallow stridents which disagree on the \texttt{$\pm$ant}$(x)$ relations. The structure \texttt{[+ant][-ant]} however, is insufficient, since consonants like \texttt{p,b,m} have that feature, and would incorrectly ban acceptable strings. To see the importance, a conventional string model must ban multiple sibilant factors {\textipa{sS,zS,sZ,zZ}}, while an unconventional model must just ban one, \texttt{$\left[\shortstack{+str\\+ant\\[-.7em]}\right]\left[\shortstack{+str\\-ant\\[-.7em]}\right]$}

Figure~\ref{fig:poset2} showcases the relationship among these structures under a precedence model $M^<$. 
The structure for \texttt{$\left[\shortstack{+str\\+ant\\[-.7em]}\right]\left[\shortstack{+str}\right]$} contains as subfactors (among others) \texttt{$\left[\shortstack{+str}\right]\left[\shortstack{+str}\right]$}, \texttt{$\left[\shortstack{+str}\right]\left[\right]$}, [], and the empty structure (not shown). 
The empty structure is a subfactor of [], and [] in turn is a subfactor of \texttt{$[\mathtt{+ant}]$} and \texttt{$[\texttt{-str}]$}, and so on. 
If the structure \texttt{$\left[\shortstack{+str\\+ant\\[-.7em]}\right]\left[\shortstack{+str\\+ant\\[-.7em]}\right]$} is grammatical, then all of its subfactors, 
are grammatical, and so are their subfactors, in turn.
Conversely, if the structure \texttt{$\left[\shortstack{+str\\+ant\\[-.7em]}\right]\left[\shortstack{+str\\-ant\\[-.7em]}\right]$} is known to be ungrammatical, then any structure which has it as a subfactor is also ungrammatical (for example, \texttt{$\left[\shortstack{+voi\\+str\\+ant\\[-2.5ex]}\right]\left[\shortstack{+str\\-ant\\[-1.5ex]}\right]$}, where the first segment is also voiced \texttt{+voi}), shown in Red in Figure~\ref{fig:poset2}.

The structure filters give the learner an advantage when confronting hypothesis spaces under a particular model. 
In particular, it allows the learner to prune vast swathes of the hypothesis space as it reaches for principal elements of features. 
If a learner identifies one structure as being grammatical, the learner may infer that all of its subfactors are also grammatical and not have to consider them. 
Alternatively, if the learner knows a structure is ungrammatical, it may infer that the ideals above it are also ungrammatical. 

Generally, these reductions can be exponential: an alphabet of size $2^n$ can be represented with $n$ unary relations in the model signature. 
However, this exponential reduction does not necessarily make learning any easier. 
The reason for this is that the size of $\Subfact_k(M,\Sigma^*)$ equals $\sum_{i=1}^k(2^n)^i$ where $n$ is the number of unary relations. 
Since a grammar is defined as a subset of $\Subfact_k(M,\Sigma^*)$, the number of considered grammars is thus very large.
Therefore, the problem of how to search this space effectively is paramount.

\section{The Learning Problem}
For some $M,k$, is $\mathcal{L}(M,k)$ learnable 
from positive data? 
The short answer is Yes \citep{Heinz-2010-SEL,Heinz-KasprzikEtAl-2012-LLHS}. 
The solution presented in these papers can be thought of as using the function $\Subfact_k(M,w)$ to identify permissible $k$-factors in words $w$ in the positive data. 
The $k$-factors that are not permissible are forbidden. 
With sufficient positive data, such a learning algorithm will converge to a grammar that generates any target language in the class.
While this solution is sound in theory, when the space of $k$-factors is very large, it is not practical.
Here, we make clear the problem the learning algorithm solves. 

We state the learning problem not in terms of converging to a correct grammar in the limit as previously studied, but instead of returning an `adequate' grammar given a finite positive sample. 
Determining what counts as an adequate grammar is what \citep{DeRaedt2008} calls a Quality Criterion.

\begin{definition}[The Learning Problem]\label{def:lp}
	Fix $\Sigma$, model $M$, and positive integer $k$. 
	For any language $L\in\mathcal{L}(M,k)$ and for any finite $D\subseteq L$, return a grammar $G$ such that 
\end{definition}
\begin{enumerate} \itemsep0em
	\item $G$ is consistent, that is, it covers the data: $D\subseteq L(G)$; 
	\item $L(G)$ is a smallest language in $\mathcal{L}$ which covers the data: so for all $L\in\mathcal{L}$ where $D\subseteq L$, we have $L(G)\subseteq L$; and
	\item $G$ includes structures $S$ that are restrictions of structures $S'$ included in other grammars $G'$ that also satisfy (1) and (2): for all $G'$ satisfying the first two criteria   
	for all $S'\in G'$, there exists $S\in G$ such that $S\sqsubseteq S'$.
\end{enumerate}

The first criterion is self-explanatory. 
The second criterion is motivated by Angluin's (1980) \nocite{angluin80b} analysis of identification in the limit.  
The third criterion requires that the grammar contain the most ``general'' subfactors. An example will help illustrate this criterion.

Consider again the grammar $G=\{M^\succ(aa),M^\succ(bb),M^\succ(c)\}$ with $\Sigma=\{a,b,c\}$. 
$L(G)$ is the same as $L(H)$ where 
$H=\{M^\succ(aa),M^\succ(bb),M^\succ(ac),M^\succ(bc), M^\succ(cc), \\M^\succ(ca),M^\succ(cb)\}$. 
In $H$ all the forbidden subfactors are of size 2, whereas $G$ encapsulates all of the 2-factors in $H$ which include $c$  with a single 1-factor $M^\succ(c)$.
Both grammars $G$ and $H$ may satisfy criteria (1) and (2) but $H$ would not satisfy criterion (3) because of $G$.

\section{A Bottom-Up Learning Algorithm} 
\label{sec:botuplearner}


\citep{DeRaedt2008} identifies two directions of inference: specific-to-general (i.e., `top-down') and general-to-specific (i.e., `bottom-up'). Since we are trying to find the most general subfactors, top-down inference has the potential to consider exponentially many more subfactors than bottom-up inference. 
It makes mores sense to traverse bottom-up, that is, from the most general subfactors possible to the most specific. 
Additionally, once a subfactor is identified as an element of the grammar, none of its superfactors (elements of its principal filter) need to be considered further.

A bottom-up learner is shown in Algorithm \ref{alg2}. 
Its input is a positive data sample $D$ and an integer $k$ that identifies the upper bound on the size of the subfactors.

\begin{algorithm} \small
	\KwData{positive sample $D$, empty structure $s_0$, \\ max constraint size $k$}
	\KwResult{$G$, a set of constraints}
	$Q \leftarrow \{s_0\};$ \\
	$G\leftarrow \emptyset;$ \\
	$V\leftarrow \emptyset;$\\
	\While{$Q \neq \emptyset$}{
		$s \leftarrow Q.\texttt{dequeue()}$;\\
		$V \leftarrow V \cup \{s\}$;\\
		\If{$\exists x \in D$ such that $s \sqsubseteq x$}{$S \leftarrow \NSupFacs(s)$;\\ 
			$S' \leftarrow \{s \in S\mid |s|\leq k \wedge (\neg\exists g\in G)[g\sqsubseteq s] \wedge s\not\in V \}$;\\
			$Q$.\texttt{enqueue}($S'$);\\}
		\Else{$G\leftarrow G\cup \{s\}$;}
	}
	\Return G;\\
	\vspace{.5em}
	\caption{Bottom-up learning algorithm for lattice-structured constraints}
	\label{alg2}
\end{algorithm}

\begin{figure}[h] 
	\centering
	\includegraphics[width=5cm]{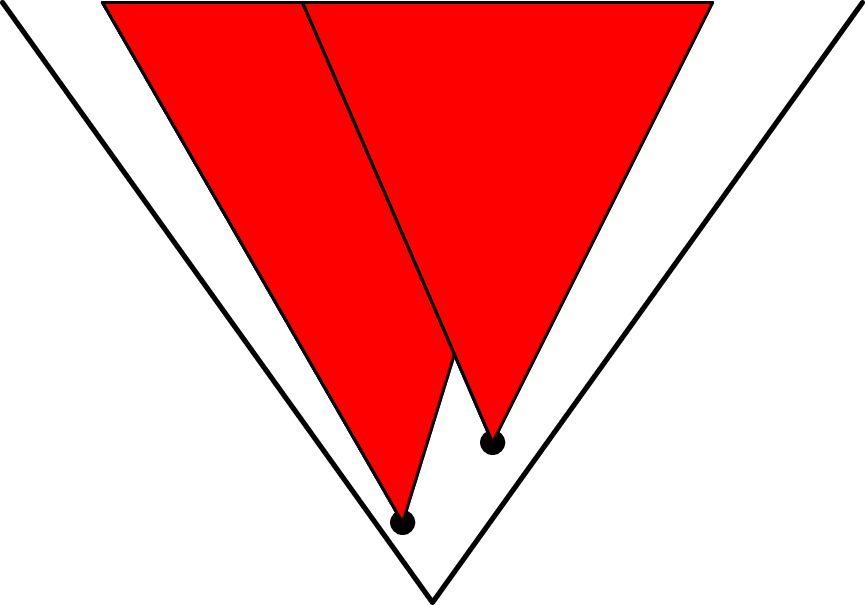}\caption{Pruning the hypothesis space}
	\label{fig:projection}
\end{figure}

The algorithm makes use of a queue $Q$, which is initialized to contain just the empty structure $s_0$. 
It also initializes two empty sets: $G$, the grammar that will ultimately be returned, and $V$, the set of `visited subfactors'. 
The subfactors in $Q$ are considered one at a time, in order, and as each subfactor $s$ is considered it is added to $V$. 
If $s$ is not a subfactor of the model of any word in the positive sample $D$ (i.e., not contained by any data point in $D$), then it is added to the grammar $G$.

If $s$ is a subfactor of the sample, it is sent to the function $\NSupFacs$, which returns a set of \emph{least} superfactors for $s$. 
For concreteness, $\NSupFacs(s)$ may be defined formally as follows: 

$ \NSupFacs(s) = \{ S\in\Subfact_k(\Sigma^*) \mid s\sqsubseteq S, \neg\exists\,S' [s\sqsubseteq S' \sqsubseteq S] \}$.

Practically $\NSupFacs$ will be defined constructively so that each subfactor in $\Subfact_k(\Sigma^*)$ is constructed only once as needed. Thus, not only will it not be needed to store the whole set $\Subfact_k(\Sigma^*)$ in memory, but the set $V$ may be excluded from the algorithm as well. 

This set of superfactors is then filtered by the following criteria: they must be smaller than $k+1$, they must contain no element of $G$ as a subfactor, and they must not have been previously considered (i.e., they cannot be in $V$). Those structures that survive this filter are added to $Q$. This procedure continues until there are no more structures left to consider in $Q$.



\begin{theorem}
	For any $L\in\mathcal{L}(M,k)$, and any finite set $P\subseteq L$ provided as input to Algorithm \ref{alg2}, it returns a grammar $G$ satisfying Definition \ref{def:lp}.
\end{theorem}

\begin{proof}
	
	Consider any $x\in D$. Algorithm 1 only adds elements to $G$ that are not subfactors of $x$, so $x\not\in\Supfact(G)$. Thus $x\in L(G)$ and $D\subseteq L(G)$, satisfying Condition (1).
	
	Consider any $L'\in\mathcal{L}$ with $D\subseteq L'$. 
	To show $L=L(G)\subseteq L'$, 
	consider any $w\in L$. 
	Then $\Subfact(w)\subseteq\Subfact(D)$ and
	$\Subfact(D)\subseteq \Subfact(L')$ since $D\subseteq L$. 
	Then $\Subfact(w)\subseteq \Subfact(L')$. 
	Hence, $w\in L'$, and so $L\subseteq L'$, satisfying Condition (2).
	
	For condition (3), we use the fact that elements in the grammar $G$ were in Q at some point. 
	Suppose $s, s'$ are subfactors such that $s \in G$, $s' \sqsubseteq s$, and ($\neg \exists x \in D)[s' \sqsubseteq M(x)]$. Since $s \in G$, then at some point $s \in Q$. 
	
	If $s' \sqsubseteq s$ then $s'$ will be added to $Q$ before $s$ is generated by $\NSupFacs$. 
	Because $Q$ is a queue, $s'$ will also be removed from $Q$ before $s$ is generated by $\NSupFacs$.
	Since $s'$ is not contained by any $M(x)$ with $x \in D$, it will be added to $G$. 
	When $s$ is generated by $\NSupFacs$, it will not pass the filter because it fails the second criterion since $s' \sqsubseteq s$ and $s' \in G$. Then $s$ is never added to $Q$, and therefore $s \notin G$, contra our original assumption. 
	Thus Condition (3) is satisfied.
\end{proof}

One aspect of the algorithm to highlight is that when a subfactor $g$ is added to $G$, it is not added to $Q$. 
Consequently, $\NSupFacs(g)$ is never added to $Q$. In this way, finding elements of $G$ prunes the remainder of the space to be searched (see figure~\ref{fig:projection}).
In general, it is not the case that every element in the principal filter of $g$ will not be generated by $\NSupFacs$ since some of these elements may belong to $\NSupFacs(x)$ for other subfactors $x$ on the $Q$. 
We expect subfactors on the `border' of $\Supfact(g)$ to be generated in this way (and then they are filtered out).
This pruning, especially when the subfactors are quite general, can significantly reduce the remaining space to be traversed.

In regard to efficiency, in the worst case, the elements of $G$ are all very specific subfactors and are greatest elements of $\Subfact_k(\Sigma^*)$. 
In this case, every subfactor  $\Subfact_k(\Sigma^*)$ will be added to $Q$ and the time complexity is thus exponential. 
However, we are primarily interested in the case when  $\Subfact_k(D)$ are a small proportion of $\Subfact_k(\Sigma^*)$. 
This constitutes an example of data sparsity.
In this case, we believe the elements of the target grammar will be much `lower' in the partial order and thus will be found much more quickly. 
Determining what conditions on $\Subfact_k(D)$ and $\Subfact_k(\Sigma^*)$ result in a polynomial time run in the size of $D$ is a focus of current research activity.

Another area of active research is developing a recipe for the $\NSupFacs$ function for models with a successor or precedence order relation and arbitary unary relations.
The basic idea underlying the bottom-up algorithm is to develop a spanning tree for the poset $\Subfact(\Sigma^*)$ and to traverse this tree in a breadth-first manner.
The function $\NSupFacs$ helps control this search. 
Ideally, $\NSupFacs$ would only generate each subfactor once, which obviates the need to store visited subfactors in $V$.
This can be accomplished to some extent in different ways. 
For incompatible unary relations, like $\mathtt{a}$ and $\mathtt{b}$ in our capitalization example, $\NSupFacs$ can be defined to prevent adding property $\mathtt{a}$ to a position that already satisfies property $\mathtt{b}$. 

For compatible unary relations, like $\mathtt{a}$ and $\mathtt{capital}$ in our capitalization example, an ordering over the unary relations such as $\mathtt{a} < \mathtt{b} < \mathtt{capital}$ can help eliminate generating the same subfactor in different ways. For example, if $\NSupFacs$ is defined to only add `lesser' unary relations to positions that already have them then it would only output [$\mathtt{capital,a}$] given the subfactor [$\mathtt{a}$] as input. On the other hand, 
when given as input the subfactor [$\mathtt{capital}$], it could not add any unary relation to this position.

\section{Conclusion} 
In this paper, we considered the problem of learning formal languages defined as the complement of the union of finitely many principal filters, whose principal elements make up the grammar.
This is one way to characterize the Strictly $k$-Local and Strictly $k$-Piecewise languages, but the generalization here lets us consider enriched representations of strings where different elements in a string can be said to share properties. it also lets us learn the shortest forbidden substrings in $SL_k$ \citep{RonEtAl1996}
This is useful in many applications where domain-specific knowledge is available and should be taken advantage of. Such enriched representations, however, have a drawback. 
The number of subfactors is large which makes identifying the principal elements of the filters difficult. 
This paper showed that the partial ordering of the subfactors motivates a bottom-up learning algorithm which finds the least subfactors whose filters do not include the positive data.


\section*{Acknowledgments}

We would like to thank James Rogers for very helpful discussion on the notion of subfactor. This work was supported by NIH grant \#R01HD87133-01 to JH.\\


\bibliography{acl2019}
\bibliographystyle{acl_natbib}

\appendix

\end{document}